\documentclass[pra,twocolumn]{revtex4}

\usepackage{amsmath}
\usepackage{amsthm}

\theoremstyle{remark}

\begin{document}

\newtheorem{cond}{Condition}
\newtheorem{prop}{Proposition}

\title{Existence of maximally correlated states}
\author{S. Camalet}
\affiliation{Sorbonne Universit\'e, CNRS, Laboratoire de Physique
 Th\'eorique de la Mati\`ere Condens\'ee, LPTMC, F-75005, 
Paris, France}

\begin{abstract}
A measure of total correlations cannot increase under deterministic 
local operations. We show that, for any number of systems, 
this condition alone does not guarantee the existence of maximally 
correlated states. Namely, there is no state that simultaneously 
maximizes all the measures satisfying it. If, in addition, the measures 
do not increase with probability unity under local measurements, 
then such states exist for two systems. They are the maximally 
entangled states. For a larger number of systems, it depends on 
their Hilbert space dimensions.
\end{abstract} 

\maketitle

\section{Introduction}

In order to apprehend quantum entanglement, many measures 
have been introduced \cite{HHHH}. Some of them have clear 
operational meanings, such as distillable entanglement or 
entanglement cost \cite{BBPSSW,BDSW,PV}, others are readily 
computable, such as negativity \cite{ZHSL,ViWe}. They all 
have an essential property that makes them proper 
entanglement quantifiers, which is the monotonicity under local 
operations and classical communication \cite{HHHH,BDSW,VPRK,V}. 
More precisely, when a multipartite state is deterministically changed 
into another using such means, the measure does not have 
a higher value for the obtained state than for the original one. 
A measure satisfying this requirement is called entanglement 
monotone. Similar monotones have been defined for other quantum 
resources, such as nonuniformity, coherence, or asymmetry 
\cite{GMNSH,NBCPJA,PCBNJ,BCP,LZYDL,BSFPW,KPTAJ,K}.

A multipartite state is said to be not more entangled than another 
when the former can be deterministically obtained from 
the latter by local operations and classical communication, or, 
equivalently, when any entanglement monotone is not larger 
for the former than for the latter. This defines a partial order. 
Two states can be incomparable, which means that, in going from 
one to the other, some entanglement monotones increase 
whereas others decrease. However, there are states that 
simultaneously minimize all the entanglement monotones. They 
are the separable states which are the mixtures of product 
states \cite{W}. Moreover, for bipartite systems, there are 
so-called maximally entangled states which are not less entangled 
than any other state on the same Hilbert space \cite{PV,LZFFL,qp}. 

The entanglement ordering of multipartite states is based on 
deterministic operations. Performing a local measurement and 
selecting a specific outcome can result in a state more or less 
entangled than the initial one, with or without classical 
communication. So, the values of an entanglement monotone for 
the state before the measurement and for the states after it are 
not necessarily related in a particular way. However, many familiar 
entanglement monotones are nonincreasing on average under local 
measurements \cite{HHHH,V}. This means that the average of 
the postmeasurement amounts of entanglement, each weighted 
with the corresponding measurement probability, is not larger 
than the premeasurement amount of entanglement. Other quantum 
resource monotones are nonincreasing on average under 
appropriate measurements 
\cite{NBCPJA,PCBNJ,BCP,LZYDL,BSFPW,KPTAJ,K}. 
 
In this paper, we address the issue of ordering multipartite states 
according to total correlations. We are, in particular, concerned 
with the possible existence of maximally correlated states, which 
is of interest for at least two reasons. First, for pure bipartite 
states, the only form of correlation is entanglement \cite{MBCPV}, 
and so maximally entangled states may be maximally correlated 
in some sense in this case. Second, it has been shown that 
the internal entanglement of a bipartite system is constrained 
by its total correlations with another system \cite{qp}. The found 
upper bound always decreases as the correlations increase, 
but, depending on the measure of correlations used, the internal 
entanglement necessarily vanishes for maximal correlations or not. 

The correlations between systems that do not influence each 
other cannot increase. The time evolution of independent 
systems is given by deterministic local operations. Thus, 
a measure of total correlations cannot increase under such 
operations \cite{HV,BM}. We name such measures as correlation 
monotones. We introduce them and show their basic properties 
in Sec.\ref{Cm}. As we will see in Sec.\ref{Sobcm}, if no other 
requirement is imposed, there is no maximally correlated state on 
any Hilbert space, in stark contrast with the case of entanglement. 
In Sec.\ref{Comcs}, it is shown that there are maximally correlated 
states for bipartite systems when only correlation monotones 
which are nonincreasing on average under local measurements 
are considered. In fact, the same result follows from a weaker 
condition, namely, that the measures do not increase with probability 
unity under local measurements. Correlation monotones that do 
not fulfill this last condition are discussed in Sec.\ref{Cmmpes}. 
In Sec.\ref{IHsd}, we consider the case of more than two systems, 
for which the existence or not of maximally correlated states 
depends on the Hilbert space dimensions of the systems. Finally, 
in Sec.\ref{C}, we summarize our results.

\section{Correlation monotones}\label{Cm}

For $N$ systems with Hilbert spaces ${\cal H}_n$, a local operation 
is characterized by Kraus operators of the form 
\begin{equation}
\tilde K_q=K_q \otimes I_1^>, \;\; I_{N}^< \otimes K_q,\;\; 
\text{or  }I^<_n \otimes K_q \otimes I_n^> , \label{Kq}
\end{equation}
where $I^<_n=\otimes_{m<n}I_m$ and 
$I^>_n=\otimes_{m>n}I_m$ with $I_n$ the identity operator 
on ${\cal H}_n$. The local operators $K_q$ map ${\cal H}_n$ into 
an Hilbert space ${\cal H}'_n$ possibly different from ${\cal H}_n$. 
They form a complete set of Kraus operators, namely, they satisfy 
$\sum_q K_q^\dag  K_q=I_n$. If ${\cal H}'_n$ is the same for 
every $q$, a deterministic operation $\Lambda$ can be defined. 
It changes the state $\rho$ of the $N$ systems into 
$\Lambda(\rho)=\sum_q \tilde K_q \rho \tilde K_q^\dag$. 
As discussed in the introduction, we consider measures $C$, 
termed correlation monotones, that obey the following Condition.
\begin{cond}\label{cond1}
The measure $C$ is nonincreasing under deterministic local 
operations, i.e., $C(\Lambda(\rho)) \le C(\rho)$ for any state 
$\rho$ and deterministic local operation $\Lambda$. 
\end{cond} 
A measure of total correlations can, for example, be a mimimal 
distance to the set of product states, 
$C(\rho)=\inf_{\{\delta_n\}_n} D(\rho,\otimes_n \delta_n)$, 
where the infimum is taken over all the density operators 
of the considered systems. Provided $D$ satisfies 
$D(\Lambda(\omega),\Lambda(\omega'))\le D(\omega,\omega')$ 
for any quantum operation $\Lambda$, $C$ is a correlation 
monotone. Some possible choices for $D$ are the relative entropy, 
the Hellinger distance, or the Bures distance 
\cite{qp,MPSVW,ACMBMAV,BCLA}. The above definition gives 
the total mutual information
\begin{equation}
I(\rho)=\sum_{n=1}^N S\big(\rho^{(n)}\big)-S(\rho) , \label{I}
\end{equation}
for the relative entropy \cite{MPSVW}. In this expression, $S$ is 
the von Neumann entropy and $\rho^{(n)}=
\operatorname{tr}_{\otimes_{m\ne n}{\cal H}_m} \rho$ is 
the state of system $n$, where $\operatorname{tr}_{\cal H}$ 
denotes the partial trace over the Hilbert space ${\cal H}$. 

Any two product states can be transformed into one another by 
local operations and so a correlation monotone $C$ assumes 
the same value for all the product states. Moreover, $C(\rho)$ 
cannot be smaller than this value since any state $\rho$ can be 
changed into a product state by local operations. It is usually set 
to zero as there is no correlation between systems in product 
states. The above mentioned measures of total correlations 
vanish for product states. Another characteristic of correlation 
monotones is that they do not depend explicitly on the Hilbert 
spaces ${\cal H}_n$, in particular on their dimensions. This follows 
from the Proposition below.
\begin{prop}\label{prop1}
Consider any $r\times r$ Hermitian matrix $M$ with trace unity, 
$r$ integer vectors ${\bf i}_s=(i_{s,n})_{n=1}^N$, and 
correlation monotone $C$.

The value $C(\rho)$ is the same for all the states 
$\rho=\sum_{s,t=1}^r M_{s,t} 
| {\bf i}_s \rangle  \langle {\bf i}_t|$
where $M_{s,t}$ denotes the elements of $M$ and 
$| {\bf i}_s \rangle=\otimes_{n=1}^N | i_{s,n} \rangle_n$ 
with $| i \rangle_n$ any orthonormal states of any Hilbert 
space ${\cal H}_n$.
\end{prop}
\begin{proof}
Consider the states $\rho_k=\sum_{s,t=1}^r M_{s,t} 
| {\bf i}_s \rangle_k {_k\langle} {\bf i}_t|$ where $k=1$ or $2$ 
and $| {\bf i}_s \rangle_k
=\otimes_{n=1}^N | i_{s,n} \rangle_{k,n}$ with 
$| i \rangle_{k,n}$ any orthonormal states of any Hilbert space 
${\cal H}_{k,n}$. Let us define the Hilbert spaces 
$\tilde {\cal H}_n={\cal H}_{1,n} \otimes {\cal H}_{2,n}$. 
The states $\rho_1$ and $\tilde \rho_1= 
\rho_1 \otimes| {\bf i}_1 \rangle_2 {_2\langle} {\bf i}_1 |$
can be changed into each other by the local transformations 
$\rho\mapsto
\rho\otimes | i_{1,n} \rangle_{2,n} {_{2,n}\langle}i_{1,n}|$ 
and $\rho\mapsto\operatorname{tr}_{{\cal H}_{2,n}}\rho$, 
and so $C(\tilde \rho_1)=C(\rho_1)$. The state $\tilde \rho_1$ 
is changed into $\tilde \rho_2= 
| {\bf i}_1 \rangle_1 {_1\langle} {\bf i}_1 | \otimes \rho_2$ 
by applying the $N$ operations with Kraus operators 
$K^{(n)}=\otimes_{m\ne n}\tilde I_m \otimes U_n$ 
where $\tilde I_n$ is the identity operator on $\tilde {\cal H}_n$ 
and $U_n$ is a unitary operator on $\tilde {\cal H}_n$ such 
that $U_n| i \rangle_{1,n}|j \rangle_{2,n} 
=| j \rangle_{1,n}| i \rangle_{2,n}$. The state $\tilde \rho_2$ 
is transformed back into $\tilde \rho_1$ by applying the $N$ 
operations with Kraus operators $(K^{(n)})^\dag$, and hence 
$C(\tilde \rho_2)=C(\tilde \rho_1)$, which leads to 
$C(\rho_2)=C(\rho_1)$ and finishes the proof.
\end{proof}
For pure bipartite states, Proposition \ref{prop1} yields 
the following. Any such state $|\psi \rangle$ can be written as
$|\psi \rangle=\sum_{i=1}^r 
\sqrt{\lambda_i} | i \rangle_1 \otimes | i \rangle_2$, where $r$ 
is its Schmidt rank, $\lambda_i$ denotes its Schmidt coefficients, 
and $| i \rangle_n$ are orthonormal states of ${\cal H}_n$. 
The corresponding $r\times r$ matrix $M$ is given by 
$M_{s,t}=\sqrt{\lambda_s\lambda_t}$ and the $r$ vectors 
${\bf i}_s$ by ${\bf i}_s=(s,s)$. These vectors are the same for all 
the states with rank $r$. Thus, $C(|\psi \rangle\langle \psi|)$ 
depends only on the Schmidt rank and coefficients of $|\psi \rangle$.

\section{State ordering based on the correlation monotones}
\label{Sobcm}

We say that $\rho$ is not more correlated than $\rho'$, according 
to Condition \ref{cond1}, if and only if $C(\rho) \le C(\rho')$ for 
any correlation monotone $C$. We denote this relation by 
$\rho \prec_1 \rho'$. It is a preorder, i.e., $\rho \prec_1 \rho$ and 
if $\rho \prec_1 \rho'$ and $\rho' \prec_1 \rho''$ then 
$\rho \prec_1 \rho''$. States $\rho$ and $\rho'$ such that 
$\rho \prec_1 \rho'$ and $\rho' \prec_1 \rho$, which means that 
$C(\rho)=C(\rho')$ for any correlation monotone $C$, are said to 
be equally correlated for $\prec_1$. This is, for instance, the case 
of two pure bipartite states with identical Schmidt coefficients. 
States $\rho$ and $\rho'$ such that $C(\rho')-C(\rho)$ is positive 
for some correlation monotones $C$ and negative for others are 
said to be incomparable for $\prec_1$. The ordering $\prec_1$ can 
also be characterized as follows.
\begin{prop}\label{prop2}
Let $\rho$ and $\rho'$ be two states on any $N$-partite Hilbert 
spaces. The three following statements are equivalent. 

\noindent i) some deterministic local operations change $\rho'$ 
into $\rho$ ,

\noindent ii) $\rho \prec_1 \rho'$ ,

\noindent iii) $C(\rho) \le C(\rho')$ for any correlation monotone 
$C$ vanishing only for product states.
\end{prop} 
\begin{proof}
We first show that (i) implies (ii). If $\rho'$ can be changed 
into $\rho$ by deterministic local operations, there are 
$N$ such operations $\Lambda^{(n)}$, each corresponding 
to a system $n$, such that 
$\rho=\circ_{m=1}^N \Lambda^{(m)} (\rho')$. 
For any $C$ obeying Condition 1, 
$C[\Lambda^{(N)}(\rho')]\le C(\rho')$ and 
$C[\circ_{m=n-1}^N \Lambda^{(m)} (\rho')]
\le C[\circ_{m=n}^N \Lambda^{(m)} (\rho')]$ 
for $n=2, \ldots, N$, and so $C(\rho)\le C(\rho')$.

The implication (ii)$\Rightarrow$(iii) follows directly from 
the definition of the preorder $\prec_1$.

We finally show that not (i) implies not (iii). Let $\rho$ and 
$\rho'$ be two states such that $\rho'$ cannot be changed 
into $\rho$ by deterministic local operations. Consequently, 
$\rho$ is not a product state. Assume first that $\rho'$ is a product 
state, and define a function $C$ as follows. For any state 
$\tilde \rho$, $C(\tilde \rho)=0$ if $\tilde \rho$ is a product state 
and $C(\tilde \rho)=c > 0$ otherwise. In particular, $C(\rho')=0$ 
and $C(\rho)=c$. Consider any state $\tilde \rho$ and 
any deterministic local operation $\Lambda$. If $C(\tilde \rho)=c$ 
then $C[\Lambda(\tilde \rho)] \le C(\tilde \rho)$ is trivially fulfilled. 
If $C(\tilde \rho)=0$, then $\tilde \rho$ is a product state, and so 
is $\Lambda(\tilde \rho)$, and hence $C[\Lambda(\tilde \rho)]=0$. 
Consequently, there is a correlation monotone $C$ vanishing only 
for product states such that $C(\rho) > C(\rho')$. 

Assume now that $\rho'$ is not a product state. Since any state can 
be transformed into any product state by local operations, 
a function $C$ can be defined as follows. For any state $\tilde \rho$, 
$C(\tilde \rho)=0$ if $\tilde \rho$ is a product state, 
$C(\tilde \rho)=c'>0$ if $\rho'$ can be changed into $\tilde \rho$ by 
local operations and $\tilde \rho$ is not a product state, and 
$C(\tilde \rho)=c > c'$ otherwise. In particular, $C(\rho')=c'$, 
$C(\rho)=c$, and $C$ vanishes only for product states. 
Consider any state $\tilde \rho$ and any deterministic 
local operation $\Lambda$. If $C(\tilde \rho)=c$ then 
$C[\Lambda(\tilde \rho)] \le C(\tilde \rho)$ is trivially fulfilled. 
If $C(\tilde \rho)=c'$, then $\rho'$ can be changed into $\tilde \rho$ 
by local operations, and hence also into $\Lambda(\tilde \rho)$, 
which gives $C[\Lambda(\tilde \rho)] \le c'$. If $C(\tilde \rho)=0$, 
then $\tilde \rho$ is a product state, and so is $\Lambda(\tilde \rho)$, 
and hence $C[\Lambda(\tilde \rho)]=0$. 
\end{proof} 
For measures of total correlations, the class of correlations 
monotones which vanish only for product states is of interest. 
An example is the total mutual information \eqref{I}. The above 
Proposition shows that adding this restriction does not change 
the ordering of the multipartite states. 

As seen above, any product state is not more correlated, according 
to Condition \ref{cond1}, than any other state. One can wonder 
whether, similarly, there are states that simultaneously maximize 
all the correlation monotones. This may be meaningful only for 
a given Hilbert space ${\cal H}=\otimes_{n=1}^N {\cal H}_n$. 
For instance, for spaces ${\cal H}_n$ of identical dimension $d$, 
the maximum value of the total mutual information \eqref{I} on 
${\cal H}$, $N \ln d$, is strictly increasing with $d$. The following 
general result can be shown.
\begin{prop}\label{prop3}
For any $N$ Hilbert spaces ${\cal H}_n$, there is no maximally 
correlated state on ${\cal H}=\otimes_{n=1}^N {\cal H}_n$ 
for $\prec_1$, i.e., no state $\rho'$ on ${\cal H}$ such that 
$\rho \prec_1 \rho'$ for any $\rho$ on ${\cal H}$.
\end{prop}
\begin{proof}
Assume there is $\rho'$ on ${\cal H}$ such that 
$\rho \prec_1 \rho'$ for any $\rho$ on ${\cal H}$. We first treat 
the case $N>2$. Define the measures $C_{\{x,y\}}$, where 
$\{x,y\} \subset \{1,\ldots,N\}$, as follows. Consider any 
$N$-partite state $\rho=\sum_{s,t} M_{s,t} 
| {\bf i}_s \rangle  \langle {\bf i}_t|$
where $| {\bf i}_s \rangle=\otimes_{n=1}^N | i_{s,n} \rangle_n$ 
with $| i \rangle_n$ any orthonormal states of any Hilbert space 
${\cal H}'_n$, and the corresponding reduced density operator 
for any two systems $x$ and $y$, 
$$\rho_{\{x,y\}}=\sum_{s,t} M_{s,t} \prod_{n\ne x,y} 
\delta_{i_{s,n} , i_{t,n}}
| i_{s,x} , i_{s,y} \rangle\langle  i_{t,x}, i_{t,y}| , $$
where $| i_{s,x} , i_{s,y} \rangle
=| i_{s,x} \rangle_x \otimes | i_{s,y} \rangle_y$,
and let $C_{\{x,y\}}(\rho)=E_f(\rho_{\{x,y\}})$, where $E_f$ is 
the entanglement of formation, which is a correlation monotone. 
The above expression for $\rho_{\{x,y\}}$ makes apparent that, 
due to Proposition \ref{prop1}, $C_{\{x,y\}}(\rho)$ depends only 
on the matrix elements $M_{s,t}$ and the integer vectors 
${\bf i}_s$. For deterministic local operations $\Lambda$ 
acting on system $x$ or system $y$, with local Kraus 
operators $K_q$, $C_{\{x,y\}}[\Lambda(\rho)]
=E_f[\tilde \Lambda(\rho_{\{x,y\}})]$ 
where $\tilde \Lambda$ is the operation with Kraus operators 
$K_q \otimes I_y$ or $I_x \otimes K_q$, respectively. For all 
the other deterministic local operations, 
$C_{\{x,y\}}[\Lambda(\rho)]=E_f(\rho_{\{x,y\}})$. 
Consequently, $C_{\{x,y\}}$ is a correlation monotone, and 
hence $C_{\{x,y\}}(\rho) \le C_{\{x,y\}}(\rho')$ for 
any $\rho$ on ${\cal H}$. 

For any $\tilde \rho$ on ${\cal H}_x \otimes {\cal H}_y$, 
the above inequality with $\rho=
\otimes_{n\ne x,y} I_n \otimes \tilde \rho/\Pi_{n\ne x,y} d_n$, 
where $d_n$ is the dimension of ${\cal H}_n$, shows that 
$E_f(\tilde \rho) \le E_f(\rho'_{\{x,y\}})$. In other words, 
$\rho'_{\{x,y\}}$ maximizes $E_f$ on 
${\cal H}_x \otimes {\cal H}_y$. Thus, assuming, without loss 
of generality, that $d_x \le d_y$, it is given by eq.\eqref{mps} 
with orthonormal states  $|i \rangle_x$ and $|q,i \rangle_y$ 
of ${\cal H}_x$ and ${\cal H}_y$, respectively \cite{LZFFL}, 
and so $\rho'=\sum_{q,q'} |\tilde q \rangle \langle \tilde q' | 
\otimes \Omega_{q,q'}$, where $\Omega_{q,q'}$ are operators 
on $\otimes_{n\ne x,y} {\cal H}_n$ such that 
$\operatorname{tr} \Omega_{q,q'}=p_q \delta_{q,q'}$. Consider 
now $\rho'_{\{x,z\}}$ where $z \ne y$. From the above expression 
for $\rho'$, it follows that $\rho'_{\{x,z\}}=d_x^{-1} I_x \otimes 
\operatorname{tr}_{\otimes_{n \ne x,y,z} {\cal H}_n} 
\sum_q \Omega_{q,q}$, and thus $E_f(\rho'_{\{x,z\}})=0$, 
which contradicts that $E_f$ is maximum on 
${\cal H}_x \otimes {\cal H}_z$ for $\rho'_{\{x,z\}}$. 

In the particular case $N=2$, since $\rho'$ maximizes 
the correlation monotone $E_f$ on ${\cal H}$, it is given by 
eq.\eqref{mps} assuming, without 
loss of generality, that $d_1 \le d_2$ \cite{LZFFL}. 
The operation with Kraus operators $\sqrt{p_q} I_1 \otimes U_q$ 
where $U_q$ is a unitary operator on ${\cal H}_2$ such that 
$|q, i \rangle_2=U_q|1,i \rangle_2$ 
transforms $|\tilde 1 \rangle \langle \tilde 1 |$ into $\rho'$, 
and so $\rho' \prec_1 |\tilde 1 \rangle \langle \tilde 1 |$. 
This implies that $|\tilde 1 \rangle\langle \tilde 1|$ is a pure 
maximally correlated state on ${\cal H}$ for $\prec_1$, which 
is not possible, due to Proposition \ref{prop4}.
\end{proof}
To prove the non-existence of maximally correlated states for 
$\prec_1$ in the particular case $N=2$, the Proposition below is 
used.
\begin{prop}\label{prop4}
Two pure bipartite states with the same Schmidt rank have identical 
Schmidt coefficients or are incomparable for $\prec_1$. 
\end{prop}
\begin{proof}
Consider two pure bipartite states $|\psi \rangle$ and 
$|\psi' \rangle$ with the same Schmidt rank $d$ such that 
$|\psi \rangle \langle \psi | \prec_1 |\psi' \rangle \langle \psi' |$. 
Denote $\lambda_i$ the Schmidt coefficients of $|\psi \rangle$. 
Define $|\varphi \rangle=\sum_{i=1}^d \sqrt{\lambda_i}
|i\rangle \otimes |i\rangle$ where $\{|i\rangle\}_i$ is 
an orthonormal basis of an Hilbert space ${\cal H}_0$ of 
dimension $d$, $|\varphi' \rangle \in {\cal H}_0\otimes{\cal H}_0$ 
whose Schmidt coefficients are those of $|\psi' \rangle$, 
$\rho=|\varphi \rangle \langle \varphi |$, and 
$\rho'=|\varphi' \rangle \langle \varphi' |$. 
From Proposition \ref{prop1}, it follows that $\rho \prec_1 \rho'$. 
Thus, due to Proposition \ref{prop2}, there are two deterministic 
local operation $\Lambda_n$, each corresponding to a system 
$n$, such that $\rho=\Lambda_1 \circ \Lambda_2 (\rho')$. 
We denote the eigenstates of $\Lambda_2(\rho')$ by 
$| \tilde p \rangle$. As $\Lambda_1$ is linear and $\rho$ is pure, 
the above equality gives 
$\rho=\Lambda_1 (|\tilde p \rangle \langle \tilde p |)$ for any $p$. 
Kraus operators of $\Lambda_1$ are of the form $K_q\otimes I_0$ 
where $I_0$ is the identity operator on ${\cal H}_0$ and the linear 
maps $K_q:{\cal H}_0\rightarrow {\cal H}_0$ are such that 
$\sum_q K^\dag_q K_q=I_0$. The state $|\tilde 1 \rangle$ can be 
written as $|\tilde 1 \rangle=\sum_{i=1}^d \sqrt{\mu_i}
|i'\rangle \otimes |i''\rangle$ where $\{ |i'\rangle \}_i$ and 
$\{ |i''\rangle \}_i$ are orthonormal bases of ${\cal H}_0$ and 
some coefficients $\mu_i$ may vanish. Since $\Lambda_1$ changes 
$|\tilde 1 \rangle \langle \tilde 1 |$ into the pure state $\rho$, 
one has, for any $q$ and $i$, 
$$\sqrt{\mu_i} K_q |i'\rangle=\sqrt{p_q} \sum_{j=1}^d 
\sqrt{\lambda_j} \langle i'' | j \rangle |j \rangle , $$ where 
$p_q=\langle \tilde 1 | K^\dag_qK_q \otimes I_0 |\tilde 1 \rangle$. 
Since $\lambda_j \ne 0$ for any $j$ and $\{ |j\rangle \}_j$ is 
a basis of ${\cal H}_0$, the above sum is nonzero for any $i$. 
Moreover, at least one $p_q$ is nonvanishing. So, $\mu_i \ne 0$ 
for any $i$, and hence $K_q=\sqrt{p_q} K$ where 
$K:{\cal H}_0\rightarrow {\cal H}_0$ is independent of $q$. 
It is necessarily a unitary operator, and so $\Lambda_1$ is 
a unitary operation, and $\Lambda_2(\rho')$ a pure state with 
Schmidt coefficients $\lambda_i$. By following the same steps as 
above, with $\Lambda_2(\rho')$, $\Lambda_2$, and 
$|\varphi'\rangle$ in lieu of, respectively, $\rho$, $\Lambda_1$, 
and $|\tilde 1 \rangle$, one can shown that $\Lambda_2$ is also 
a unitary operation and the Schmidt coefficients of $|\psi' \rangle$ 
are $\lambda_i$.
\end{proof}
As an example, consider the pure states $|\psi_\epsilon \rangle$ 
with Schmidt rank $2$ and coefficients $\epsilon$ and $1-\epsilon$. 
For any $\eta>0$, which can be chosen as small as wished, there is 
a correlation monotone $C$ such that 
$C(|\psi_\eta \rangle \langle \psi_\eta |)
>C(|\psi_{\epsilon} \rangle \langle \psi_{\epsilon} |)$ 
for any $\epsilon \ne \eta$. 

\section{Correlation orderings with maximally correlated states}
\label{Comcs}

We now examine specific classes of correlation monotones $C$ for 
which $C(|\psi_\epsilon \rangle \langle \psi_\epsilon |)$ is 
a nondecreasing function of $\epsilon$ on $[0,1/2]$. They are 
defined by the Conditions below. A local measurement is 
characterized by a complete set of Kraus operators $\tilde K_{q,s}$ 
of the form of eq.\eqref{Kq}. When the measurement is performed 
on systems in the state $\rho$, the probability $p_q$ of outcome 
$q$ and the corresponding postmeasurement state $\rho_q$ are 
given by
\begin{equation}
p_q=\sum_{s=1}^{t_q}
\operatorname{tr}(\rho \tilde K^\dag_{q,s}\tilde K_{q,s}) , \;\;\;
\rho_q=\sum_{s=1}^{t_q} \tilde K_{q,s} \rho \tilde K^\dag_{q,s} 
/ p_q . \label{prho}
\end{equation} 
The number $t_q$ of terms in the above sums can depend on $q$. 
If $t_q=1$, the measurement is said to be efficient \cite{NC}. 
Note that here the codomains of the local operators $K_{q,s}$ 
and $K_{q',s'}$ can be different from one another provided 
$q' \ne q$. A very commun requirement for measures of 
entanglement or of any quantum resource 
\cite{HHHH,V,NBCPJA,PCBNJ,BCP,LZYDL,BSFPW,KPTAJ,K} 
reads here as follows. 
\begin{cond}\label{cond2}
The measure $C$ is nonincreasing on average under local 
measurements, i.e., $\sum_q p_q C(\rho_q) \le C(\rho)$ for 
any state $\rho$ and complete set of Kraus operators \eqref{Kq}, 
where $p_q$ and $\rho_q$ are given by eq.\eqref{prho}.
\end{cond}
We also consider the following weaker requirement.
\begin{cond}\label{cond3}
The measure $C$ does not increase with probability unity 
under local measurements, i.e., $\min_q C(\rho_q) \le C(\rho)$ for 
any state $\rho$ and complete set of Kraus operators \eqref{Kq}, 
where $\rho_q$ is given by eq.\eqref{prho}.
\end{cond}
Clearly, if $C$ satisfies Condition \ref{cond2} then it also satisfies 
Condition \ref{cond3}. For correlation monotones, it 
is enough to impose Condition \ref{cond2} or Condition 
\ref{cond3} for efficient projective measurements, as shown by 
the Proposition below.
\begin{prop}\label{prop5}
A correlation monotone $C$ fulfills Condition \ref{cond2} 
(Condition \ref{cond3}) if and only if $C$ fulfills it for efficient local 
measurements whose Kraus operators \eqref{Kq} are projectors. 
\end{prop}
\begin{proof}
Let $\rho$ be any state on any multipartite Hilbert space 
$\otimes_n {\cal H}_n$. It is enough to consider a measurement 
characterized by a complete set of Kraus operators of the form 
$\tilde K_{q,s}=K_{q,s} \otimes I_1^>$ 
where $K_{q,s} : {\cal H}_1 \rightarrow {\cal H}_1^{(q)}$, 
$q=1, \ldots, d$, and $s=1, \ldots, t_q$. Let ${\cal H}'_1$ be 
an Hilbert space of dimension larger than those of the spaces 
${\cal H}_1^{(q)}$, and define 
$K'_{q,s} : {\cal H}_1 \rightarrow {\cal H}'_1$ by 
$\langle j' | K'_{q,s} | i \rangle
={_q\langle} j | K_{q,s} | i \rangle$ for $j$ not larger than 
the dimension of ${\cal H}_1^{(q)}$ and 
$\langle j' | K'_{q,s} | i \rangle=0$ otherwise, where 
$\{ | i \rangle \}_i$, $\{ | i' \rangle \}_i$, and 
$\{ | i \rangle_q\}_i$ are orthonormal bases of ${\cal H}_1$, 
${\cal H}'_1$, and ${\cal H}^{(q)}_1$, respectively, 
and let $\tilde K'_{q,s}=K'_{q,s} \otimes I_1^>$. 
The operators $K'_{q,s}$ satisfy 
$\sum_{q,s}(K'_{q,s})^\dag K'_{q,s}=I_1$. Let ${\cal H}_1''$ 
be an Hilbert space of dimension $d$ and $\{ | q \rangle \}_q$ 
be an orthonormal basis of this space. The deterministic 
local operation $\Lambda$ with Kraus operators 
$| q \rangle \otimes \tilde K'_{q,s}$ changes $\rho$ into 
$\Lambda(\rho)=\sum_q A_q$ where $A_q=| q \rangle \langle q | 
\otimes \sum_s \tilde K'_{q,s} \rho (\tilde K'_{q,s})^\dag$. 
Performing the efficient projective local measurement with Kraus 
operators $(| q \rangle \langle q | \otimes I'_1) \otimes I_1^>$ 
where $I'_1$ is the identity operator on ${\cal H}'_1$ on 
$\Lambda(\rho)$ gives $A_q/p_q$ with probability 
$p_q=\operatorname{tr}(A_q)=\sum_s 
\operatorname{tr}(\tilde K_{q,s}^\dag\tilde K_{q,s} \rho)$. 
From Proposition \ref{prop1}, it follows that $C(A_q/p_q)=C(\rho_q)$ 
where $\rho_q=\sum_s \tilde K_{q,s} \rho \tilde K^\dag_{q,s} 
/ p_q$. Finally, $C[\Lambda(\rho)] \le C(\rho)$ leads to the results. 
The converses are trivial.
\end{proof}
Using this Proposition, it can be shown that the total mutual 
information \eqref{I} satisfies Condition \ref{cond2}.
\begin{proof}
The total mutual information can be written as $I(\rho)
=\min_\delta S(\rho|| \delta)$ where $S(\rho||\delta)$ is 
the quantum relative entropy and the mininum is taken over 
all the product states on the same Hilbert space than $\rho$ 
\cite{MPSVW}. For any states $\rho$ and $\delta$ on 
any Hilbert space ${\cal H}$ and Kraus operators $\tilde K_{q,1}$ 
such that $\sum_q \tilde K_{q,1}^\dag \tilde K_{q,1}$ is the identity 
operator on ${\cal H}$, it can be shown that 
$S(\rho|| \delta) \ge \sum_q p_q S(\rho_q|| \delta_q)$ where 
$p_q$ and $\rho_q$ are given by eq.\eqref{prho} with $t_q=1$, 
and $\delta_q$ is given by a similar expression with $\delta$ 
in place of $\rho$ \cite{BCP}. 
If ${\cal H}=\otimes_{n=1}^N {\cal H}_n$, 
$\delta$ is a product state on ${\cal H}$, and the operators 
$\tilde K_{q,1}$ are of the form of eq.\eqref{Kq}, 
then the density operators $\delta_q$ are also product states. 
Thus, the above expression for the total mutual information 
gives $S(\rho|| \delta) \ge \sum_q p_q I(\rho_q)$, which 
leads to the result.
\end{proof}
It is usual to derive the analogue 
of Condition \ref{cond1} for measures of entanglement or of 
other quantum resources assuming they are convex and obey 
the analogue of Condition \ref{cond2} \cite{V,BCP}. 
These assumptions lead here to very particular correlation 
monotones, as shown by the Proposition below. 
\begin{prop}\label{prop6}
For any $N$, if a function of the $N$-partite states is convex and 
nonincreasing on average under local measurements, then it is 
an entanglement monotone.
\end{prop}
\begin{proof}
Consider a function $C$ of the $N$-partite states convex and 
nonincreasing on average under local measurements. Let us 
first show that $C$ is a correlation monotone. Let $\rho$ be 
any $N$-partite state and $\Lambda$ any deterministic local 
operation with Kraus operators $\tilde K_{q,1}$. The convexity 
of $C$ yields $C[\Lambda(\rho)] \le \sum_q p_q C(\rho_q)$ 
where $p_q$ and $\rho_q$ are given by eq.\eqref{prho} 
with $t_q=1$. So, since $C$ satisfies Condition \ref{cond2}, 
it also satisfies Condition \ref{cond1}. 
An entanglement monotone is nonincreasing under state 
transformations involving only local operations and classical 
communication. Such a transformation is a sequence of maps 
involving only local operations and one-way classical 
communication. Thus, it is enough to show that 
$C[\Lambda'(\rho)] \le C(\rho)$ for any map $\Lambda'$ 
given by $\Lambda'(\rho)=\sum_q \Lambda^{(q)}
(\tilde K_{q,1} \rho \tilde K_{q,1}^\dag)$ 
with a complete set of Kraus operators 
$\tilde K_{q,1}=K_q \otimes I_1^>$ and deterministic local 
operations $\Lambda^{(q)}$ whose Kraus operators are of 
the form $I_1 \otimes K_{q'}^{(q)} \otimes I_2^>$. 
Since $C$ is convex and obeys Condition 1, and $\Lambda^{(q)}$ 
is linear, $C[\Lambda'(\rho)] \le \sum_q p_q C(\rho_q)$ where 
$p_q$ and $\rho_q$ are given by eq.\eqref{prho} with $t_q=1$. 
The wished inequality follows immediately from Condition \ref{cond2}. 
\end{proof}
We remark that, as a consequence of this Proposition, for 
any measure $C$ fulfilling Condition \ref{cond2}, there is 
an entanglement monotone that does not exceed $C$, which is 
the convex hull of $C$ \cite{PRL}. Condition \ref{cond3}, together 
with Condition \ref{cond1}, leads to the following result.
\begin{prop}\label{prop7}
If $C$ is a correlation monotone obeying Condition \ref{cond3}, 
then, for any pure bipartite state $|\psi \rangle$ with Schmidt rank 
$r$ and coefficients $\lambda_i$, 
$C(|\psi \rangle \langle \psi |)=s[(\lambda_i)_{i=1}^r]$ where 
$s$ is an entropy, i.e., 
$s({\boldsymbol \lambda}) \le s({\boldsymbol \mu})$ when 
the vector ${\boldsymbol \lambda}$ majorizes the vector 
${\boldsymbol \mu}$ \cite{MOA}.
\end{prop}
\begin{proof}
Due to Proposition \ref{prop1}, for any pure bipartite state 
$\rho=|\psi\rangle\langle \psi|$, 
$C(\rho)=s({\boldsymbol \lambda})$ where 
${\boldsymbol \lambda}=(\lambda_i)_{i=1}^{r}$ with 
the Schmidt rank $r$ and coefficients $\lambda_i$ of 
$|\psi \rangle$. Consider any vector ${\boldsymbol \lambda}'$ 
majorized by ${\boldsymbol \lambda}$ \cite{MOA}. 
Any pure state $\rho'$ with Schmidt coefficients $\lambda'_i$ 
can be changed into $\rho$ by local operations and one-way 
classical communication, i.e.,  $\rho
=\sum_q \Lambda^{(q)}(\tilde K_{q,1} \rho' \tilde K_{q,1}^\dag)$ 
with a complete set of Kraus operators 
$\tilde K_{q,1}=K_q \otimes I_2$ and deterministic local operations 
$\Lambda^{(q)}$ whose Kraus operators are of the form 
$I_1 \otimes K_{q'}^{(q)}$ \cite{N}. Since $\rho$ is pure 
and $\Lambda^{(q)}$ linear, $\rho=\Lambda^{(q)}(\rho'_q)$ 
for any $q$, where $\rho'_q$ is given by eq.\eqref{prho} 
with $t_q=1$ and $\rho'$ in place of $\rho$. Thus, 
$C(\rho) \le C(\rho'_q)$ for any $q$, and so $C(\rho) \le C(\rho')$, 
which can be rewritten 
as $s({\boldsymbol \lambda}) \le s({\boldsymbol \lambda}')$.
\end{proof}
As is well known, a pure bipartite state $|\psi \rangle$ is not more 
entangled than another one $|\psi' \rangle$ if and only if the vector 
made up of the Schmidt coefficients of $|\psi \rangle$ majorizes 
that corresponding to $|\psi' \rangle$ \cite{N}. Consequently, due 
to Proposition \ref{prop7}, a correlation monotone fulfilling 
Condition \ref{cond3} necessarily coincides with an entanglement 
monotone for pure bipartite states. 

Similarly to the ordering $\prec_1$, we define the preorders 
$\prec_2$ with correlation monotones fulfilling Condition 
\ref{cond2} and $\prec_3$ with correlation monotones fulfilling 
Condition \ref{cond3}. These three orderings are related by 
$\rho \prec_1 \rho' \Rightarrow \rho \prec_3 \rho' 
\Rightarrow \rho \prec_2 \rho'$. Contrary to $\prec_1$, for 
$\prec_2$ and $\prec_3$, there are maximally correlated states on 
bipartite Hilbert spaces. This results from the following Proposition.
\begin{prop}\label{prop8}
Let $\rho$ be a state on any bipartite Hilbert space ${\cal H}$. 
The three following statements are equivalent.

\noindent i) $\rho$ is maximally entangled on ${\cal H}$ ,

\noindent ii) $\rho$ is maximally correlated on ${\cal H}$ 
for $\prec_3$ ,

\noindent iii) $\rho$ is maximally correlated on ${\cal H}$ 
for $\prec_2$. 
\end{prop}
\begin{proof}
Assume, without loss of generality, that the dimension $d_1$ of 
${\cal H}_1$ is not larger than the dimension $d_2$ of ${\cal H}_2$. 

We first show that (i) implies (ii). Let $\rho'$ be a pure maximally 
entangled state on ${\cal H}={\cal H}_1\otimes{\cal H}_2$, i.e., 
with Schmidt coefficients $1/d_1$. Consider any state $\rho$ on 
${\cal H}$ and denote its eigenvalues by $\mu_p$ and 
the corresponding eigenstates by $| p \rangle$. Let us introduce 
the Hilbert spaces ${\cal H}'_2$ and 
$\tilde {\cal H}_2={\cal H}_2 \otimes {\cal H}'_2$. Provided 
the dimension of ${\cal H}'_2$ is not smaller than $d_1d_2$, 
$\rho$ can be expressed as 
$\rho=\operatorname{tr}_{{\cal H}'_2}| \psi \rangle \langle \psi |$ 
where $| \psi \rangle=\sum_{p=1}^{d_1d_2} \sqrt{\mu_p} 
| p \rangle \otimes | p' \rangle$ with orthonormal states 
$| p' \rangle$ of ${\cal H}'_2$. The state $| \psi \rangle$ can also 
be written as $| \psi \rangle=\sum_{i=1}^r \sqrt{\lambda_i} 
| i \rangle_1 \otimes | i \rangle_2$ where $| i \rangle_1$ and 
$| i \rangle_2$ are orthonormal states of ${\cal H}_1$ and 
$\tilde {\cal H}_2$, respectively, and $r \le d_1$. The vector 
$(\lambda_i)_{i=1}^r$ majorizes $(1/d_1, \ldots,1/d_1)$ 
\cite{MOA}. Thus, $\rho'$ can be changed into the pure state 
$\tilde \rho=| \psi \rangle \langle \psi |$ by local operations and 
one-way classical communication \cite{N}. Consequently, for 
any measure $C$ fulfilling Conditions 1 and 3, 
$C(\tilde \rho) \le C(\rho')$, see the proof of Proposition \ref{prop7}, 
which means $\tilde \rho \prec_3 \rho'$. From the relation
between the orderings $\prec_1$ and $\prec_3$ and the fact 
that $\operatorname{tr}_{{\cal H}'_2}$ is a local operation on 
 $\tilde {\cal H}_2$, it follows that $\rho \prec_3 \tilde \rho$, 
and so $\rho \prec_3 \rho'$. If there are mixed maximally entangled 
states on ${\cal H}$, they can be transformed into pure maximally 
entangled states on ${\cal H}$ by local operations, as shown below, 
and are hence maximally correlated on ${\cal H}$ for $\prec_3$. 

The implication (ii)$\Rightarrow$(iii) directly follows from the fact 
that any two states $\rho$ and $\rho'$ such that $\rho \prec_3 \rho'$ 
necessarily satisfy $\rho \prec_2 \rho'$. 

Consider a maximally correlated state $\rho'$ on ${\cal H}$ for 
$\prec_2$ and a maximally entangled state $\rho$ on ${\cal H}$.
For any correlation monotone $C$ obeying Condition 2, 
$C(\rho) \le C(\rho')$. The entanglement of formation $E_f$ is such 
a measure, and so $E_f(\rho')=E_f(\rho)$. Consequently, $\rho'$ is 
given by eq.\eqref{mps} \cite{LZFFL}, and is thus maximally 
entangled on ${\cal H}$, as shown below.
\end{proof}
A state $\rho'$ is maximally entangled on ${\cal H}$ if and only if 
$E(\rho) \le E(\rho')$ for any entanglement monotone $E$ and 
state $\rho$ on ${\cal H}$. Such states exist for bipartite Hilbert 
spaces. When the dimension $d_1$ of ${\cal H}_1$ is not larger 
than the dimension $d_2$ of ${\cal H}_2$, which can always be 
assumed without loss of generality, they are necessarily of the form 
\begin{equation}
\rho=\sum_{q=1}^Q p_q |\tilde q \rangle \langle \tilde q | 
\text{ with } |\tilde q \rangle=
\sum_{i=1}^{d_1} |i \rangle_1 \otimes |q,i \rangle_2/\sqrt{d_1} ,
\label{mps}
\end{equation} 
where $|i \rangle_1$ and $|q,i \rangle_2$ are orthonormal states 
of ${\cal H}_1$ and ${\cal H}_2$, respectively, and the probabilities 
$p_q$ sum to unity. This follows from the fact that the only 
states that maximise the entanglement of formation on ${\cal H}$ 
are given by eq.\eqref{mps} \cite{LZFFL}. If $d_2 \ge 2d_1$, 
$\rho$ can be mixed. Reciprocally, any state \eqref{mps} is 
maximally entangled on ${\cal H}$ as it can be transformed into 
the pure maximally entangled state $|\tilde 1 \rangle$ by 
the deterministic local operation with Kraus operators 
$I_1 \otimes \sum_{i=1}^{d_1} |1,i \rangle_2 {_2\langle} q,i|$ 
where $q=1, \ldots, Q$ 
and, if $d_2/d_1 \ne Q$, $I_1 \otimes (I_2-\sum_{q=1}^Q \Pi_q)$ 
where $\Pi_q=\sum_{i=1}^{d_1} |q,i \rangle_2 {_2\langle} q,i|$. 
Note that this operation changes a state given by eq.\eqref{mps} 
with any Schmidt coefficients $\lambda_i$, in place of $1/d_1$, 
into a pure state. As a result of Proposition \ref{prop8}, 
the maximally correlated states on ${\cal H}$ for $\prec_2$ and 
$\prec_3$ are also the density operators \eqref{mps}. 

\section{Correlation monotones maximum for partially 
entangled states}\label{Cmmpes}

If a bipartite correlation monotone $C$ is not maximum, on some 
Hilbert space, for the maximally entangled states, then, due to 
Proposition \ref{prop8}, it does not satisfy Condition \ref{cond3}. 
This means that, for at least one state $\rho$, there is a local 
measurement such that $C(\rho_q) > C(\rho)$ for every outcome 
$q$, where $\rho_q$ denotes the postmeasurement states given 
by eq.\eqref{prho}. This is actually the case for all the maximally 
entangled states which do not maximise $C$, as we show now. 
Such a state $\rho$ on ${\cal H}$ is given by eq.\eqref{mps}. 
Assume there are states on ${\cal H}$ for which $C$ is larger 
than $C(\rho)$. Any state is not more correlated, according to 
Condition 1, than a pure state on the same Hilbert space, see 
the proof of Proposition \ref{prop8}. So, there are pure states 
$|\psi\rangle$ of ${\cal H}$ with Schmidt coefficients 
$\lambda_i \ne 1/d_1$ such that 
$C(|\psi\rangle\langle \psi|)=c>C(\rho)$. The local measurement 
with Kraus operators 
$\tilde K_q = \sum_{i=1}^{d_1} \sqrt{\lambda_i} 
|i \rangle_1 {_1\langle} (i+q-1)\mod d_1 +1 | \otimes I_2$ 
where $q=1,\ldots,d_1$, performed on $\rho$, gives 
$C(\rho_q)=c$ for every outcome $q$. 

Interesting examples of bipartite correlation monotones that do 
not obey Condition \ref{cond3} can be defined from Bell inequalities. 
The relation between Bell nonlocality and entanglement is 
not straightforward \cite{W,BCPSW,PRA}. A Bell inequality reads 
$\sum_{s,t,x,y} \beta_{x,y}^{s,t} p(s,t|x,y) \le b$ 
where $\beta_{x,y}^{s,t}$ are real coefficients and $p(s,t|x,y)$ 
is the probability of the outcomes $s$ and $t$ of the local 
measurements $x$ and $y$, respectively. When 
the set of probabilities is Bell local, the left side can approach but 
not exceed $b$. The largest value of this left side, for a given 
state $\rho$, is $$B (\rho) = \sup_{( F^{(n)}_{s|x} )_{n,s,x}} 
\sum_{s,t,x,y} \beta_{x,y}^{s,t} \operatorname{tr} 
\big(\rho F^{(1)}_{s|x} \otimes F^{(2)}_{t|y}\big) ,$$
where the supremum is taken over the positive operators 
$F^{(n)}_{s|x}$ such that $\sum_s F^{(n)}_{s|x}=I_n$. 
The measure $B$ satisfies Condition \ref{cond1} since, 
for any complete set of Kraus operators 
$K_q : {\cal H}_n \rightarrow {\cal H}'_n$ 
and positive operators $F'_s$ on ${\cal H}'_n$ summing to 
the corresponding identity operator, the operators 
$F_s=\sum_q K_q^\dag F'_s K_q$ are positive and sum to $I_n$. 
If $B (\rho)>b$, the Bell inequality can be violated by $\rho$. 
For some Bell inequalities, $B$ is maximum for partially entangled 
states \cite{BCPSW} and so does not fulfill Condition \ref{cond3}. 
It follows from the above discussion that a maximal violation of 
the Bell inequality can nevertheless be achieved with a maximally 
entangled state $\rho$ by first performing a local measurement 
on it, and for every outcome of the measurement. Note that 
this does not even suppose that $B (\rho)>b$. In the context 
of Bell nonlocality, performing a local measurement and selecting 
an outcome is termed filtering.

\section{Influence of the Hilbert space dimensions}\label{IHsd}

For more than two systems, the existence or not of maximally 
correlated states for $\prec_2$ or $\prec_3$ depends on 
the dimensions of the Hilbert spaces ${\cal H}_n$. The two 
following Propositions can be proved.
\begin{prop}\label{prop9}
Consider any $N$ Hilbert spaces ${\cal H}_n$ of respective 
dimensions $d_n$ and denote by ${\cal F}$ the set of 
the nonempty subsets ${\cal E}$ of $\{1,\ldots,N\}$ such 
that $\Pi_{n\in {\cal E}} d_n \le \sqrt{d}$ where 
$d=\Pi_{n=1}^N d_n$.

If $\rho$ is a maximally correlated state on 
$\otimes_{n=1}^N {\cal H}_n$ for $\prec_2$ or $\prec_3$, 
then, for any ${\cal E} \in {\cal F}$, 
$\operatorname{tr}_{\otimes_{n\notin {\cal E}} {\cal H}_n} \rho$ 
is the maximally mixed state on 
$\otimes_{n\in {\cal E}} {\cal H}_n$. 

If, moreover, 
$\max_{{\cal E} \in {\cal F}} \Pi_{n\in {\cal E}} d_n>\sqrt{d/2}$, 
then $\rho$ is pure.
\end{prop}
\begin{proof}
Consider any $N$ Hilbert spaces ${\cal H}_n$ and any nonempty 
set ${\cal E}\subset \{1,\ldots,N\}$, and define the measure 
$C_{\cal E}$ as the entanglement of formation $E_f$ for 
the bipartition 
${\cal H}={\cal H}^{\in}\otimes{\cal H}^{\notin}$ 
where ${\cal H}=\otimes_n {\cal H}_n$, 
${\cal H}^{\in}=\otimes_{n\in {\cal E}} {\cal H}_n$ 
and ${\cal H}^{\notin}=\otimes_{n\notin {\cal E}} {\cal H}_n$. 
It is a correlation monotone fulfilling Condition \ref{cond2}, 
since $E_f$ is, and operators \eqref{Kq} can be rewritten as 
$\tilde K_q=I^{\in} \otimes K_q$ 
or $\tilde K_q=K_q \otimes I^{\notin}$ where $K_q$ acts on 
${\cal H}^{\notin}$ or ${\cal H}^{\in}$, respectively, and 
$I^{\in}$ and $I^{\notin}$ are the identity operators on 
${\cal H}^{\in}$ and ${\cal H}^{\notin}$, respectively. 
Assume there is a maximally correlated state $\rho$ on ${\cal H}$ 
for $\prec_2$ or $\prec_3$. It maximizes all the monotones 
$C_{\cal E}$ on ${\cal H}$. If ${\cal E}\in{\cal F}$, 
the dimension $d^{\in}$ of ${\cal H}^{\in}$ is not larger than 
the dimension $d^{\notin}=d/d^{\in}$ of ${\cal H}^{\notin}$. 
Thus, $\rho$ is given by eq.\eqref{mps} with $d_1$ replaced 
by $d^{\in}$, the states $|i\rangle_1$ by orthonormal 
states of ${\cal H}^{\in}$, and the states $|q,i\rangle_2$ by 
orthonormal states of ${\cal H}^{\notin}$ \cite{LZFFL}. 
Consequently, $\operatorname{tr}_{{\cal H}^{\notin}} \rho$ 
is the maximally mixed state on ${\cal H}^{\in}$. Moreover, 
$\rho$ is necessarily pure if $d^{\notin} <2 d^{\in}$. 
\end{proof}
For identical dimensions $d_n$, the pure states with the property 
stated in the above Proposition are known as absolutely maximally 
entangled states \cite{HS,S,HCLRL,HGS}. Furthermore, in this case, 
the above condition on the dimensions is always satisfied for 
even $N$ and never for odd $N$. For $d_n=2$, absolutely maximally 
entangled states do not exist when $N=4$ \cite{HS} and $N \ge 7$ 
\cite{S,HGS}. So, due to Proposition \ref{prop9}, for $N$ two-level 
systems, there is no maximally correlated state neither for $\prec_2$ 
nor for $\prec_3$ when $N$ is even and different from $6$. 
Using the results of Ref.\cite{MV}, such states can be shown to exist 
on the Hilbert spaces considered in the Proposition below.
\begin{prop}
Let ${\cal H}_n$ be any $N$ Hilbert spaces of respective 
dimensions $d_n$. If $d_N \ge d$ where 
$d=\Pi_{n=1}^{N-1} d_n$, then there are maximally correlated 
states on ${\cal H}=\otimes_{n=1}^N {\cal H}_n$ for $\prec_2$ 
and $\prec_3$, which are also maximally entangled on ${\cal H}$ 
and are of the form
$$\sum_{i_1, \ldots, i_{N-1} =1}^{d_1, \ldots, d_{N-1}} 
\otimes_{n=1}^{N-1} |i_n \rangle_n \otimes 
| i_1, \ldots, i_{N-1} \rangle_N/\sqrt{d} , $$
where $|i \rangle_n$ and $| i, j, \ldots \rangle_N$ are orthonormal 
states of ${\cal H}_n$ with $n<N$ and ${\cal H}_N$, respectively.
\end{prop}
\begin{proof}
Consider any state $\rho$ on ${\cal H}
=\otimes_{n=1}^N {\cal H}_n$. 
Let us introduce the Hilbert spaces ${\cal H}'_N$ and 
$\tilde {\cal H}_N={\cal H}_N \otimes {\cal H}'_N$. Provided 
the dimension of ${\cal H}'_N$ is not smaller than that of ${\cal H}$, 
$\rho$ can be expressed as 
$\rho=\operatorname{tr}_{{\cal H}'_N}| \psi \rangle \langle \psi |$, 
where $| \psi \rangle 
\in \otimes_{n=1}^{N-1} {\cal H}_n \otimes \tilde {\cal H}_N$, 
see the proof of Proposition 8. The state $| \psi \rangle$ can always 
be written as
$$| \psi \rangle
=\sum_{i_1, \ldots, i_{N-1} =1}^{d_1, \ldots, d_{N-1}} 
\lambda_{ i_1, \ldots, i_{N-1}} 
\otimes_{n=1}^{N-1} |i_n \rangle_n \otimes 
| i_1, \ldots, i_{N-1} \rangle_N , $$
where $|i \rangle_n$ are orthonormal states of ${\cal H}_n$, 
and $| i, j, \ldots \rangle_N$ are normalized states of 
$\tilde {\cal H}_N$. We denote by $\tilde {\cal H}'_N$ a subspace 
of $\tilde {\cal H}_N$ of dimension $d=\Pi_{n=1}^{N-1}d_n$ that 
contains the states $| i, j, \ldots \rangle_N$. From the relation 
between the orderings $\prec_1$ and $\prec_3$ and the fact 
that $\operatorname{tr}_{{\cal H}'_N}$ is a local operation on 
$\tilde {\cal H}_N$, it follows that $\rho \prec_3 \tilde \rho$ 
where $\tilde \rho=| \psi \rangle \langle \psi |$. 

Consider $\rho'=| \psi' \rangle \langle \psi' |$  and 
$\tilde \rho'=| \tilde \psi' \rangle \langle \tilde \psi' |$ where 
$|\psi'\rangle$ and $|\tilde \psi'\rangle$ are given by the above 
expression with $\lambda_{ i_1, \ldots, i_{N-1}}=1/\sqrt{d}$ 
and orthonormal states $| i, j, \ldots \rangle_N$ of ${\cal H}_N$ 
and $\tilde {\cal H}'_N$, respectively. The states $\rho'$ and 
$\tilde \rho'$ are equally correlated for $\prec_1$, and so also 
for $\prec_3$, see Proposition \ref{prop1}. There is a set of Kraus 
operators $M_q = \otimes_{n=1}^{N-1} U_q^{(n)} \otimes K_q$ 
where $U_q^{(n)}$ denotes a unitary operator on ${\cal H}_n$ 
and the linear maps 
$K_q:\tilde{\cal H}'_N\rightarrow\tilde{\cal H}'_N$ are such that 
$\sum_q K_q^\dag K_q=\tilde I_N$ with $\tilde I_N$ the identity 
operator on $\tilde{\cal H}'_N$ for which 
$| \psi \rangle=\sqrt{d}M_q |\tilde \psi' \rangle$ for any $q$ 
\cite{MV}. This equality can be rewritten as $\tilde \rho=U_q 
\otimes \tilde I_N \tilde \rho'_q U_q^\dag \otimes \tilde I_N$ 
where $U_q=\otimes_{n=1}^{N-1} U_q^{(n)}$ 
and $\tilde \rho'_q$ is given by eq.\eqref{prho} with $t_q=1$, 
$\tilde K_{q,1}=I_N^< \otimes K_q$, and $\tilde \rho'$ in place 
of $\rho$. Thus, for any correlation monotone fulfilling 
Condition \ref{cond3}, $C(\tilde \rho)=C(\tilde \rho'_q)$ for 
any $q$, and so $C(\tilde \rho) \le C(\tilde \rho')$. This means 
$\tilde \rho \prec_3 \tilde \rho'$, and hence $\rho \prec_3 \rho'$ 
and, due to the relation between the orderings $\prec_3$ and 
$\prec_2$, $\rho \prec_2 \rho'$.
\end{proof}

\section{Conclusion}\label{C}

In summary, we have shown that there is no maximally 
correlated states when the measures of total correlations are only 
required to be correlation monotones. In contrast, there exist 
such states if, in addition, they do not increase with probability 
unity under local measurements. In this case, there are always 
maximally correlated states for two systems, which are 
the corresponding maximally entangled states. For a larger 
number of systems, two opposite examples have been considered. 
We have furthermore seen that a correlation monotone fulfilling 
the above mentioned condition necessarily coincides 
with an entanglement monotone for pure bipartite states. 
Bell inequalities maximally violated by partially 
entangled states have also been discussed. Correlation monotones 
that can increase with probability unity under local measurements 
have been defined from them. It would be of interest to determine 
whether the existence of maximally correlated states can follow 
from a condition weaker than that considered here.


\begin{thebibliography}{99}

\bibitem{HHHH} R. Horodecki, P. Horodecki, M. Horodecki, and 
K. Horodecki, Quantum entanglement, Rev. Mod. Phys. {\bf 81}, 
865 (2009).

 \bibitem{BBPSSW} C.H. Bennett, G. Brassard, S. Popescu, 
B. Schumacher, J. Smolin, and W. K. Wootters, Purification of 
Noisy Entanglement and Faithful Teleportation via Noisy Channels, 
Phys. Rev. Lett {\bf 76}, 722 (1996).

\bibitem{BDSW} C.H. Bennett, D. P. DiVincenzo, J. Smolin, and 
W. K.Wootters, Mixed State Entanglement and Quantum Error 
Correction, Phys. Rev. A {\bf 54}, 3824 (1996).

\bibitem{PV} M. B. Plenio and S. Virmani, An introduction to 
entanglement measures, Quantum Inf. Comput. {\bf 7}, 1 
(2007). 

\bibitem{ZHSL} K. \. Zyczkowski, P. Horodecki, A. Sanpera and 
M. Lewenstein, Volume of the set of separable states, Phys. 
Rev. A {\bf 58}, 883 (1998).

\bibitem{ViWe} G. Vidal and R.F. Werner, Computable 
measure of entanglement, Phys. Rev. A {\bf 65}, 032314 
(2002).

\bibitem{VPRK} V. Vedral, M. B. Plenio, M. A. Rippin, and 
P. L. Knight, Quantifying Entanglement, Phys. Rev. Lett {\bf 78}, 
2275 (1997).

\bibitem{V} G. Vidal, Entanglement monotones, J. Mod. Opt. 
{\bf 47}, 355 (2000).

\bibitem{GMNSH} G. Gour, M. P. M\" uller, 
V. Narasimhachar, R. W. Spekkens, and 
N. Y. Halpern, The resource theory of 
informational nonequilibrium in thermodynamics, 
Phys. Rep. {\bf 583}, 1 (2015).

\bibitem{BCP} T. Baumgratz, M. Cramer, 
and M.B. Plenio, Quantifying Coherence, 
Phys. Rev. Lett. {\bf 113}, 140401 (2014).

\bibitem{NBCPJA} C. Napoli, T. R. Bromley, M. Cianciaruso, 
M. Piani, N. Johnston, G. Adesso, Robustness of coherence: 
An operational and observable measure of quantum coherence, 
Phys. Rev. Lett. {\bf 116}, 150502 (2016).

\bibitem{PCBNJ} M. Piani, M. Cianciaruso, T. R. Bromley, 
C. Napoli, N. Johnston, G. Adesso, Robustness of asymmetry 
and coherence of quantum states, 
Phys. Rev. A {\bf 93}, 042107 (2016).

\bibitem{LZYDL} C. L. Liu, D.-J. Zhang, X.-D. Yu, Q.-M. Ding, 
L. Liu, A new coherence measure based on fidelity, 
Quantum Inf. Process. {\bf 16}, 198 (2017).

\bibitem{BSFPW} K. Bu, U. Singh, S.-M. Fei, A. K. Pati, J. Wu, 
Max-relative entropy of coherence: an operational coherence 
measure, Phys. Rev. Lett. {\bf 119}, 150405 (2017).

\bibitem{KPTAJ} H. Kwon, C.-Y. Park, K. C. Tan, D. Ahn, 
H. Jeong, Coherence, Asymmetry, and Quantum Macroscopicity, 
Phys. Rev. A {\bf 97}, 012326 (2018).

\bibitem{K} N. K. Kollas, Optimization free measures of quantum 
resources, Phys. Rev. A {\bf 97}, 062344 (2018).

\bibitem{W} R.F. Werner, Quantum states with 
Einstein-Podolsky-Rosen correlations admitting a hidden-variable 
model, Phys. Rev. A {\bf 40}, 4277 (1989).

\bibitem{LZFFL} Z.-G. Li, M.-J. Zhao, S.-M. Fei, H. Fan, and 
W. M. Liu, Mixed maximally entangled states, Quantum Inf. 
Comput. {\bf 12}, 0063 (2012).

\bibitem{qp} S. Camalet, Internal entanglement and external 
correlations of any form limit each other, Phys. Rev. Lett. 
{\bf 121}, 060504 (2018).

\bibitem{MBCPV} K. Modi, A. Brodutch, H. Cable, T. Paterek, 
and V. Vedral, The classical-quantum boundary for correlations: 
Discord and related measures, Rev. Mod. Phys. {\bf 84}, 
1655 (2012). 

\bibitem{HV} L. Henderson and V. Vedral, Classical, quantum 
and total correlations, J. Phys. A {\bf 34}, 6899 (2001).

\bibitem{BM} A. Brodutch and K. Modi, Criteria for measures 
of quantum correlations, Quantum Inf. Comput. {\bf 12}, 0721
 (2012). 

\bibitem{MPSVW} K. Modi, T. Paterek, W. Son, V. Vedral, 
M. Williamson, Unified View of Quantum and Classical 
Correlations, Phys. Rev. Lett. {\bf 104}, 080501 (2010).

\bibitem{ACMBMAV} K. M. R. Audenaert, J. Calsamiglia,
 R. Munoz-Tapia, E. Bagan, L. Masanes, A. Acin, and 
F. Verstraete, Discriminating States: The Quantum Chernoff 
Bound, Phys. Rev. Lett. {\bf 98}, 160501 (2007).

\bibitem{BCLA} T. R. Bromley, M. Cianciaruso, R. Lo Franco, 
G. Adesso, Unifying approach to the quantification of bipartite 
correlations by Bures distance,  J. Phys. A: Math. Theor. 
{\bf 47}, 405302 (2014).

\bibitem{NC} M. A. Nielsen and I. L. Chuang, {\it Quantum 
Computation and Quantum Information} (Cambridge University 
Press, 2000).

\bibitem{PRL} S. Camalet, Monogamy inequality for any local 
quantum resource and entanglement, Phys. Rev. Lett. {\bf 119}, 
110503 (2017).

\bibitem{MOA} A. W. Marshall, I. Olkin, and B.C. Arnold, 
{\it Inequalities: Theory of Majorization and its Applications}, 
Second edition, Springer Series in Statistics (Springer, New York, 
2011). 

\bibitem{N} M.A. Nielsen, Conditions for a Class of Entanglement 
Transformations, Phys. Rev. Lett. {\bf 83}, 436 (1999).

\bibitem{BCPSW} N. Brunner, D.Cavalcanti, S. Pironio, 
V. Scarani, and S. Wehner, Bell nonlocality, Rev. Mod. Phys. 
{\bf 86}, 419 (2014).

\bibitem{PRA} S. Camalet, Measure-independent anomaly of 
nonlocality, Phys. Rev. A {\bf 96}, 052332 (2017).

\bibitem{HS} A. Higuchi and A. Sudbery, How entangled can 
two couples get ?, Phys. Lett. A {\bf 273}, 213 (2000).

\bibitem{S} A. J. Scott, Multipartite entanglement, 
quantum-error-correcting codes, and entangling power 
of quantum evolutions, Phys. Rev. A {\bf 69}, 052330 (2004).

\bibitem{HCLRL} W. Helwig, W. Cui, J. I. Latorre, A. Riera, 
and H.-K. Lo, Absolute maximal entanglement and quantum
secret sharing, Phys. Rev. A {\bf 86}, 052335 (2012).

\bibitem{HGS} F. Huber, O. G\"uhne, and J. Siewert, Absolutely 
maximally entangled states of seven qubits do not exist, 
Phys. Rev. Lett. {\bf 118}, 200502 (2017).

\bibitem{MV} A. Miyake and F. Verstraete, Multipartite 
entanglement in 2 x 2 x n quantum systems, Phys. Rev. A 
{\bf 69}, 012101 (2004).

\end{thebibliography}
\end{document}